\documentclass[a4paper]{article}
\usepackage{amsfonts,amssymb,amsmath,amsthm,cite}
\usepackage{ucs} 
\usepackage[utf8x]{inputenc}
\usepackage{graphicx}
\usepackage[english]{babel}
\usepackage{slashed}
\usepackage{textcomp}
\usepackage{comment}

\evensidemargin=\oddsidemargin
\newtheorem{assumption}{Assumption}[section]
\newtheorem{theorem}{Theorem}

\newtheorem{corollary}[assumption]{Corollary}

\newtheorem{lemma}[assumption]{Lemma}

\newtheorem{definition}{Definition}

\usepackage[normalem]{ulem}

\begin{document}
\begin{center}
	{\Large\textbf{Special Functions for Heat Kernel Expansion}}
	\vspace{0.5cm}
	
	{\large A.~V.~Ivanov$^\dag$ and N.~V.~Kharuk$^\ddag$}
	
	\vspace{0.5cm}
	
	$^\dag$$^\ddag${\it St. Petersburg Department of Steklov Mathematical Institute of
		Russian Academy of Sciences,}\\{\it 27 Fontanka, St. Petersburg 191023, Russia}\\
	$^\dag${\it Leonhard Euler International Mathematical Institute, 10 Pesochnaya nab.,}\\
	{\it St. Petersburg 197022, Russia}\\
	{\it E-mail: regul1@mail.ru}\\
	$^\ddag${\it ITMO University, St.Petersburg 197101, Russia}\\
	{\it E-mail: natakharuk@mail.ru}	
\end{center}
\vskip 10mm
\date{\vskip 20mm}
\vskip 10mm
\begin{abstract}

In this paper, we study an asymptotic expansion of the heat kernel for a Laplace operator on a smooth Riemannian manifold without a boundary at enough small values of the proper time. The Seeley--DeWitt coefficients of this decomposition satisfy a set of recurrence relations, which we use to construct two function families of a special kind. 
Using these functions, we study the expansion of a local heat kernel for the inverse Laplace operator.
We show that the new functions have some important properties. For example, we can consider the Laplace operator on the function set as a shift one. 
Also we describe various applications useful in theoretical physics and,
in particular, we find a decomposition of Green's functions near the diagonal in terms of new functions.

\end{abstract}

\tableofcontents  
\newpage

\section{Introduction}	
The heat kernel method first appeared in Ref. \cite{1} and  since then plays an important role in theoretical and mathematical physics. The range of applications of this approach is very wide and currently includes thousands of works. For numerous examples, we refer the reader to the works \cite{108,110,102,109,2,3}, which detail history of the issue and various applications.

As is known, a closed formula for a solution of a heat equation in general case does not exist due to technical difficulties. In this connection, as a rule, the main work with the heat kernel is focused on studying its asymptotic expansion both near zero, see Refs. \cite{10,8,7,9}, and at infinity, see Refs. \cite{4,5,6}. As an example, we can recall finding the asymptotics of an operator trace in the case of non-zero mass, see Refs. \cite{5,10,102}.

On manifolds without  boundary, a natural ansatz for finding the asymptotic expansion at small values of the proper time is the series, see formula (\ref{K_0}), in which the functions $a_n(x,y)$   are called Seeley--DeWitt (or Hadamard, Minakshisundaram \cite{111}, and Gilkey \cite{8}) coefficients, see \cite{110,10}. The point is that the Seeley--DeWitt coefficients obey the system of recurrence differential equations \cite{10}. This property is very remarkable and is used to find the trace parts of the coefficients \cite{11}, as well as in different proofs, for example, in the Atiyah--Patodi--Singer theorem \cite{12}.

In our work we use the recurrence relations in a slightly different context. It turns out that, based on the Seeley--DeWitt coefficients, one can define a family of functions of a special form, and due to the property described above the new functions are transformed into each other by the action of the Laplace operator. Thus we obtain a set of chains (see (\ref{cep1}), (\ref{cep}), (\ref{cep2}), and (\ref{se1})), which are closely related to the construction of asymptotic decompositions and various other applications.

The main purpose of our work is to study the connection between the following four objects: a local part of the heat kernel, the Seeley--DeWitt coefficients, the family of functions, and a local fundamental solution of the Laplace operator. We show the relation between asymptotic expansion of the local part of the heat kernel and a Green's function in terms of the new family of functions.

We believe that our work will be useful in theoretical and mathematical physics. Examples include loop calculations \cite{16}, a recently proposed approach to studying the fermion number \cite{17,18}, as well as application in anomaly theory, in which the same integrals arise as in Section \ref{apl}. The field of use may also include a study of integer powers of the Laplace operator \cite{107}.

Let us shortly describe the content of the paper. In Section \ref{sec:pro} we formulate the background, give basic information about the heat kernel method, and briefly introduce the Hankel transform within the framework of the work with an exponential operator, see (\ref{g19}) and (\ref{gl9}). 

Then, in Section \ref{sec:spec}, we define a family of $\Psi$-functions, the main building blocks of which are Seeley--DeWitt coefficients, and prove the key Lemma \ref{th}. It turns out that the functions can be connected not only by the action of the Laplace operator, but also by a partial derivative with respect to a parameter. We also consider two special cases of $\Psi$-functions for integer and half-integer  index values. 

In Section \ref{sec:odd} we formulate and prove Theorem \ref{th2} on the relation between the asymptotic expansion of the local part of the heat kernel and a special family of fundamental solutions for the odd-dimensional case. In other words, we found the Hankel transform of the local part of the heat kernel in terms of well-studied Seeley--DeWitt coefficients. We also provide some examples of limit cases, when a mass parameter goes to zero. We discuss degeneracy and special cases.

Section \ref{sec:even} is devoted to the similar study for the even-dimensional case and contains several parts. First, in Section \ref{sec:even:rep} we give a derivation of the asymptotic expansion of a fundamental solution of the Laplace operator near the diagonal in terms of $\Psi$-functions. Then, we introduce an additional family of $\Phi$-functions in terms of which we are going the represent the Hankel transform of the local part of the heat kernel. Further, in  Section \ref{sec:even:res}, we study the Hankel transform and prove Theorem \ref{th3}, Lemma \ref{th4}, and Corollary \ref{th5}. The presence of several propositions is a consequence of the fact that the heat kernel can be represented as two parts, each of which solves the heat equation. In Section \ref{ex2} we discuss limits and special cases. 

In Section \ref{apl} we consider different applications, such as Green's function representation, cutoff regularization, and an integral calculation. Conclusion section consists of some comments and remarks.

\section{Problem statement}
\label{sec:pro}
\subsection{Basic concepts of heat kernel method}
Let $M$ be a $d$-dimensional compact Riemannian manifold without a boundary. Points of the manifold we notate by letters $x, y$, and $z$. Then, let us consider an open convex set $U\subset M$, so we suppose that $x,y,z\in U$. This means that all further calculations are performed locally in $U$. The smooth metric tensor is equal to $g^{\mu\nu}(x)$ locally, where $\mu,\nu\in\{1,\ldots,d\}$. Moreover, let $V$ be a Hermitian vector bundle over $M$, so by $B_\mu(x)$ we notate smooth components of a Yang--Mills connection 1-form.

Let us formulate a problem, the solution $K^A(x,y;\tau)$ of which is called the heat kernel,
\begin{equation}
	\label{tepl}
	\begin{cases}
		(\partial_\tau+A(x)+m^2)K^A(x,y;\tau)=0;
		\\
		K^A(x,y;0)=g^{-1/2}(x)\delta(x-y),
	\end{cases}
\end{equation}
where $m$ is a positive constant mass parameter and $A(x)$ is a Laplace operator with smooth coefficients. In local coordinates it has the following form
\begin{equation}
	\label{op}
	A(x)=-g^{-1/2}(x)D_{x^\mu} g^{1/2}(x)g^{\mu\nu}(x)D_{x^\nu}-v(x).
\end{equation}

Here $D_{x^\mu}=\partial_{x^\mu}+B_\mu(x)$ is the covariant derivative, $v(x)$ is a smooth potential, and $g(x)$ is the metric tensor determinant. We know from the general theory that the Laplace operator has a discrete spectrum and can have only one accumulation point at $+\infty$, see Ref. \cite{103}.
Moreover, we will further assume that all eigenvalues of $A+m^2$ are positive.

Let us study a solution for problem (\ref{tepl}) by using a suitable ansatz. For the small enough values of the proper time  $\tau$ the heat kernel is given by the following asymptotic series
\begin{equation}\label{K_0}
	K^A(x,y,\tau)\stackrel{\tau\to+0}{\sim}
	K(x,y,\tau):=\frac{\Delta^{1/2}(x,y)}{(4\pi\tau)^{d/2}}e^{-\sigma(x,y)/2\tau-\tau m^2}\sum_{k=0}^{+\infty}\tau^ka_k(x,y),
\end{equation}
where $a_k(x,y)$, $k\geqslant0$, are the Seeley--DeWitt coefficients \cite{10,8}, $\sigma(x,y)$ is the Synge’s world function \cite{104}. Note, that if in the domain $U$ the space is flat,  $g_{\mu\nu}(x)=\delta_{\mu\nu}$, then the Synge’s world function $\sigma(x,y)$ equals $|x-y|^2/2$. Then, $\Delta(x,y)$ is the Van-Vleck--Morette determinant \cite{105}, which is defined by the formula
\begin{equation}
	\Delta(x,y)=\left(g(x)g(y)\right)^{-1/2}\det\left(-\frac{\partial^2\sigma(x,y)}{\partial x^\mu\partial y^\nu}\right).
\end{equation}
Also, we have introduced the notation $K(x,y,\tau)$ for the series on the right hand side of formula (\ref{K_0}). When $\tau\to+0$, it can be considered as the asymptotic series in powers of the proper time $\tau$. In all other cases, when $\tau$ is finite or $\tau\to+\infty$, the function $K(x,y,\tau)$ should be understood as the formal series with the Seeley--DeWitt coeffitients. It is quite easy to verify, that the function near the each Seeley--DeWitt coefficient is smooth, bouded, and exponentially decreasing for all $x,y\in U$ such that $x\neq y$.

For convenience, we introduce some useful notations
\begin{equation}
	\sigma_\mu(x,y)=\partial_{x^\mu}\sigma(x,y)\,\,\,\,\mbox{and}\,\,\,\,\sigma^\mu(x,y)=g^{\mu\nu}(x)\sigma_\nu(x,y),
\end{equation}
then the following relations hold
\begin{equation}
	g^{\mu\nu}(x)\sigma_\mu(x,y)\sigma_\nu(x,y)=2\sigma(x,y)
\end{equation}
and
\begin{equation}
	\label{det}
	2\sigma^\nu(x,y)\partial_{x^\nu}\Delta^{1/2}(x,y)=d\Delta^{1/2}(x,y)
	-\Big(g^{-1/2}(x)\partial_{x^\mu} g^{1/2}(x)g^{\mu\nu}(x)\partial_{x^\nu}\sigma(x,y)\Big)\Delta^{1/2}(x,y).
\end{equation}

Using the latter properties, after substituting the expansion (\ref{K_0}) for the heat kernel into the heat equation  (\ref{tepl}), we obtain the following recurrence relations
\begin{equation}
	\label{rec}
	\begin{cases}
		\sigma^\mu(x,y) D_\mu a_0(x,y)=0,\,\,a_0(x,x)=1;\\
		(k+1+\sigma^\mu D_\mu)a_{k+1}(x,y)=-\Delta^{-1/2}(x,y)A(x)\Delta^{1/2}(x,y)a_k(x,y),\,\,k\geqslant0,
	\end{cases}
\end{equation}
that define the Seeley--DeWitt coefficients, introduced above in formula (\ref{K_0}). Additionally, we assume, that the domain $U$ is chosen small enough, so that the Synge’s world function and the Seeley--DeWitt coefficients exist as smooth functions of two variables $x,y\in U$.

Now we also should draw attention to the fact that the Synge’s world function and the Seeley--DeWitt coefficients have a local nature. It means that they can be constructed inside the domain $U$ without using information about the manifold from $M\setminus U$. This leads to the essential difference between $K^A(x,y,\tau)$ and $K(x,y,\tau)$, because the first one is the global object. The asymptotic equality imply, that the difference $K^A(x,y,\tau)-K(x,y,\tau)$, when $\tau\to+0$, can contain an exponentially small term, which carries information about the whole manifold. Despite this the both functions satisfy problem (\ref{tepl}).

\subsection{Hankel transform}
\label{subsec:han}
Let $\lambda$ and $\phi_{\lambda}$ be an eigenvalue and an eigenfunction of the operator $A+m^2$, respectively. This means that $(A+m^2)\phi_\lambda=\lambda\phi_\lambda$. Using the assumptions described above we obtain $\lambda>0$. In this case the heat kernel has the following form
\begin{equation}
	\label{g30}
	K^A(x,y;\tau)=\sum_{\lambda}e^{-\tau\lambda}\phi_\lambda^{\phantom{*}}(x)\phi_\lambda^*(y),
\end{equation}
and the last eigenfunctions satisfies the completeness and the orthogonality relations
\begin{equation}
\label{g22}
\sum_{\lambda}\phi_\lambda^{\phantom{*}}(x)\phi_\lambda^*(y)=\mathbf{1}(x,y)
,\,\,\,
\int_Md^dy\,\phi_{\lambda_1}^*(y)g^{1/2}(y)\phi_{\lambda_2}^{\phantom{*}}(y)
=\delta_{\lambda_1,\lambda_2},
\end{equation}
where $\mathbf{1}(x,y)=g^{-1/2}(x)\delta(x-y)$.

Also, we introduce the Green's function for the operator $A(x)+m^2$ by the equality
\begin{equation}
	\label{eqgreen}
	G^A(x,y)=\sum_{\lambda}\lambda^{-1}\phi_\lambda^{\phantom{*}}(x)\phi_\lambda^*(y),\,\,\,
	\mbox{that satisfies}\,\,\,\,
	(A(x)+m^2)G^A(x,y)=\mathbf{1}(x,y).
\end{equation}

Then, we consider a problem similar to (\ref{tepl}), but with the inverse operator
\begin{equation}\label{eqN}
	\begin{cases}
		\partial_\tau N^A(x,y;\tau)+\int_Md^dz\,G^A(x,z)g^{1/2}(z)N^A(z,y;\tau)=0;\\
		N^A(x,y;0)=\mathbf{1}(x,y).
	\end{cases}
\end{equation}

For convenience let us turn to a solution of an alternative problem by acting the operator $A(x)+m^2$ on the left hand side of relation (\ref{eqN}). Then the problem reduces to the form
\begin{equation}\label{eqAN}
	(\partial_\tau(A(x)+m^2)+1)N^A(x,y;\tau)=0,\,\,\,\,\,\,
	N^A(x,y;0)=\mathbf{1}(x,y).
\end{equation}

It is quite easy to show that the formula holds
\begin{equation}
	\label{gl8}
	N^A(x,y;\tau)=\sum_{\lambda}e^{-\tau/\lambda}\phi_\lambda^{\phantom{*}}(x)\phi_\lambda^*(y).
\end{equation}

For small values of the parameter $\tau$ we can rewrite the last formula in the following form
\begin{equation}
\label{g20}
N^A(x,y;\tau)=\mathbf{1}(x,y)-\tau\, G^A_{\phantom{0}}(x,y)+\sum_{k=2}^{+\infty}\frac{(-\tau)^k}{k!}
G^A_k(x,y),
\end{equation}
where we have used the completeness of the eigenfunctions (\ref{g22}), definition (\ref{eqgreen}), and notation for the $k$-fold convolution of the Green's function with itself
\begin{equation}
\label{g21}
G^A_{k+1}(x,y)=
\sum_{\lambda}\lambda^{-(k+1)}\phi_\lambda^{\phantom{*}}(x)\phi_\lambda^*(y)\,\,\,\mbox{for}\,\,\,k\in\mathbb{N}.
\end{equation}
Here we want to pay attention on the fact, that series (\ref{eqgreen}) for the Green's function converges, see estimates for the spectral parameter in \cite{103}, and, moreover,  the smoothness properties of the functions $G^A_{k}$ from (\ref{g21}) get better with increasing the parameter $k$.

Let us note, that due to the orthogonality of the eigenfunctions we can investigate operator transforms separately for each eigenvalue. Using the integration of the exponential
\begin{equation}
	\label{gl7}
	\lambda^{k}=\int_{\mathbb{R}_+}ds\,\frac{s^{k-1}}{(k-1)!}e^{-s/\lambda}
\end{equation}
and two properties of the Bessel function $J_1(s)$
\begin{equation}\label{Bess}
	-\frac{s}{2}J_1(s)=\sum_{k=1}^{+\infty}\frac{(-s^2/4)^k}{k!(k-1)!},\,\,\,\,\,\,
	\int_{\mathbb{R}_+}ds\,J_1(s)=1,
\end{equation}
after the summation of (\ref{gl7}) by $k$ we obtain
\begin{equation}\label{g19}
	e^{-\lambda\tau}=1+\sum_{k=1}^{+\infty}\int_{\mathbb{R}_+}\frac{ds}{s}\,\frac{(-\tau s)^k}{k!(k-1)!}e^{-s/\lambda}=-\int_{\mathbb{R}_+}ds\,\sqrt{\frac{\tau}{s}} J_1(2\sqrt{\tau s})(e^{-s/\lambda}-1).
\end{equation}

Then, using the set of formulae (\ref{g30}), (\ref{gl8}), and (\ref{g19}), we get the following relation between straight and inverse heat kernels
\begin{align}\label{gl9}
	K^A(x,y;\tau)=-\int_{\mathbb{R}_+}ds\,\sqrt{\frac{\tau}{s}} J_1(2\sqrt{\tau s})(N^A(x,y;s)-\mathbf{1}(x,y)).
\end{align}

Such kind of relation is quite remarkable, because it gives the way to investigate Green's function decompositions. Especially we note that the transformation with the kernel $J_1(s)$ is the Hankel transform. It is very well known, see Ref. \cite{13}, and sometimes it is used to find solutions for differential equations, see Ref. \cite{14}.

Now we want to describe shortly the objects, to which the rest of the paper is devoted. Indeed, they are connected with the study of relation (\ref{gl9}) in terms of a family of new functions, but it is a very uninformative description. To understand the structure of the paper we draw the following diagram:
\begin{equation}
\label{tabll}
\begin{tabular}{ l c l }
$K^A(x,y;\tau)$&$\stackrel{\tau\to+0}{\sim}$&$K(x,y;\tau)$\\
$\,\updownarrow$ \footnotesize{H.tr.}&&$\,\updownarrow$ \footnotesize{H.tr.}\\
$N^A(x,y;\tau)$&&$N(x,y;\tau)$\\
$\,\updownarrow$&&$\,\updownarrow$\\
$G^A_{\phantom{0}}(x,y)$&&$G(x,y)$, $\{G_{k}(x,y)\}_{k\geqslant 2}$
\end{tabular}
\end{equation}

This scheme presents the properties described above and the ones, that will be studied further. First of all, at the top we have relation (\ref{K_0}) between the heat kernel $K^A(x,y;\tau)$ and the function $K(x,y;\tau)$, when $\tau\to+0$. Then, under $K^A(x,y;\tau)$ we can see the kernel $N^A(x,y;\tau)$. They are connected by the Hankel transform in the form (\ref{gl9}). And then, we see the Green's function $G^A(x,y)$, that can be obtained from $N^A(x,y;\tau)$, and vice versa. At the same time, all objects with the index $A$ are global and can be considered at any point of the manifold $M$. Thus, the left part of the diagram reflects the general and well-known relations from the theory of operators \cite{106}.

Now we would like to draw attention to the right hand side of the scheme, with which we are going to work. First of all, we need to note that $K(x,y;\tau)$ can be considered not only as asymptotic series for $\tau\to+0$, but also as the formal series by the Seeley--DeWitt coefficients. As it will be shown in our paper, the Seeley--DeWitt coefficients can be used as a "basis functions", which are not mixed during the transformations by the parameter $\tau$. We have used the quotation marks, because the Seeley--DeWitt coefficients do not form basis in the standard sense. We will identify the equality of formal series by the symbol $=$, because this does not confuse. The equalities of other type we will comment.

Returning to the description of the right hand side of the scheme, we can see, that we are going to obtain the Hankel transform of the function $K(x,y;\tau)$. The resulting object we notate by $N(x,y;\tau)$. As it is easy to verify, $N(x,y;\tau)$ satisfies the differential equation from (\ref{eqgreen}) in the domain $U$, but not the integral equation (\ref{eqN}). It is connected with the fact, that the information about the manifold from $M\setminus U$ gets lost, and only local data, from $U$, remains. In the same way we can consider the kernel $N(x,y;\tau)$ for small values of the proper time, and obtain coefficient near the first degree of $\tau$, which we notate by $G(x,y)$, and the coefficient near $\tau^k$, where $k\geqslant 2$, which we notate by $G_{k}(x,y)$. Thus, all the objects on the right hand side inherit only local information in $U$.

Unlike the left side, where we had on the bottom line only one Green's function, on the right side additionally we have the set $\{G_{k}(x,y)\}_{k\geqslant 2}$, because they can not be represented as a $k$-fold convolution of $G(x,y)$. But they are connected by the differential operator in the following way 
\begin{align}\label{g23}
\big(A(x)+m^2\big)G_{k+1}(x,y)=G_{k}(x,y)\,\,\,\mbox{for}\,\,\,k>1,\,\,\,\mbox{and}\,\,\,
\big(A(x)+m^2\big)G_{2}(x,y)=G(x,y).
\end{align}

After the description of the scheme we can formulate the main aim of our work. We are going to define functions of a special type, that allows us to describe the decompositions for functions from the right hand side of the scheme. Such kind of decompositions are quite remarkable, and have a number of applications in theoretical physics, see Section \ref{apl}.

We are going to call $G(x,y)$ a Green's function instead of the local fundamental solution. We hope this does not mislead the reader, because in the rest of the paper we have only objects without $A$.

Also in the following we are going to omit the arguments $x,y\in U$ of the functions mentioned above, except in cases where it is necessary. Therefore, we have the following reductions:
\begin{equation}
	\label{de}
	A=A(x),\,\,
	a_j=a_j(x,y),\,\,
	\Delta=\Delta(x,y),\,\,
	\sigma=\sigma(x,y),
\end{equation}
\begin{equation}
	\label{de1}
	\mathbf{1}=\mathbf{1}(x,y),\,\,
	K(\tau)=K(x,y;\tau),\,\,
	N(s)=N(x,y;s).
\end{equation}

\section{Family of $\Psi$-functions}
\label{sec:spec}
\begin{definition}
Let $\alpha\in\mathbb{C}$, $\omega\in\mathbb{R}_+$, and $g^{\alpha}=\{g_k^{\alpha}\}_{k=0}^{+\infty}$ be a set of complex numbers, depending on the parameter $\alpha$. Then we define the following special function of two variables $x,y\in U$, corresponding to the operator $A$, by the formula
\begin{equation}
\label{def2}
\Psi_\alpha^{\omega}[g^{\alpha}]=\Delta^{1/2}\sum_{k=0}^{+\infty}\frac{(-1)^kg_k^{\alpha}}{2^k}\omega^{k-\alpha}a_k.
\end{equation}
\end{definition}

In general case the last series can be understood as the formal one.  If we have the condition $|\omega|\to+0$, then the series has the asymptotic behaviour. In some special cases of the Seeley--DeWitt coefficients and $g^{\alpha}$ the series converges.
\begin{lemma}
	\label{th}
In addition to the conditions described above let us require, that the following additional relations hold
	\begin{equation}
		\label{soot1}
		g_{k-1}^{\alpha}=g^\alpha_k(k-\alpha),\,\,\,\,\,\,
		g^{\alpha+1}_k=-2g^\alpha_{k-1},\,\,\,\,\,\,
		k\geqslant 1.
	\end{equation}
	
	Then we have
	\begin{equation}\label{Psi_rec_1}
		-2\partial_\omega\Psi_\alpha^{\omega}[g^{\alpha}]=\Psi_{\alpha+1}^{\omega}[g^{\alpha+1}],
	\end{equation}
	\begin{equation}
		\label{Psi_rec_2}
		\Delta^{\frac{1}{2}}\sigma^\mu D_\mu\Delta^{-\frac{1}{2}}\Psi_{\alpha}^{\omega}[g^{\alpha}]=
		-A\Psi_{\alpha-1}^{\omega}[g^{\alpha-1}]-\alpha\Psi_{\alpha}^{\omega}[g^{\alpha}]+
		\frac{\omega}{2}\Psi_{\alpha+1}^{\omega}[g^{\alpha+1}].
	\end{equation}
	
	If we additionally assume, that the inequality $\mathrm{Re}(\alpha)<k+d/2-1$ is satisfied for all $k$, such that $g_k^{\alpha}\neq0$, then we have
	\begin{equation}\label{Psi_rec}
		A\Psi_\alpha^{\sigma}[g^{\alpha}]=(d/2-\alpha-1)\Psi_{\alpha+1}^{\sigma}[g^{\alpha+1}],
	\end{equation}
	where $\sigma$ is the Synge's world function.
\end{lemma}
\noindent\textbf{Proof:}
The first relation (\ref{Psi_rec_1}) follows from the differentiation by the parameter $\sigma$ and applying formulae (\ref{soot1}). The second one can be verified by using formulae (\ref{rec}), (\ref{soot1}), and (\ref{Psi_rec_1}), and using the equality $k\omega^{k-\alpha}=\omega^{-\alpha}\partial_\omega \omega^\alpha\omega^{k-\alpha}$.

Then, let us apply the operator $A$  to the product $\Delta^{1/2}\sigma^\beta a_n$ with $\beta>1-d/2$ and $n\geqslant0$. We get
\begin{equation}
	\label{soot}
	A\left(\Delta^{1/2}\sigma^\beta a_n\right)=-\sigma^\beta\Delta^{1/2}\left(n+1+\sigma^\mu D_\mu\right)a_{n+1}
	-\beta\sigma^{\beta-1}\Delta^{1/2}\left(d+2(\beta-1)+2\sigma^\mu D_\mu\right)a_n,     
\end{equation}
where we have used the relations from (\ref{rec}). Substituting this equality and making a change of variables, after application of the relations from (\ref{soot1}) we obtain the last statement of the lemma. $\square$

\subsection{The case $\alpha\in\mathbb{Z}$}
Let us consider an example of functions from definition (\ref{def2}), where the parameter $\alpha$ takes only integer values.
It is quite obvious that the set of coefficients $g^{\alpha}=\{g^{\alpha}_k\}_{k=0}^{+\infty}$ is uniquely defined by the relations from (\ref{soot1}) and one non-zero value.
\begin{lemma}
	\label{cel1}
	Let $g^0_0=1$, then, taking into account the equalities from (\ref{soot1}), we have 
	\begin{equation}
		\label{cel}
		g^p_k=\frac{(-2)^p}{\Gamma(k-p+1)}\,\,\,\mbox{for all}\,\,\,
		k\geqslant 0\,\,\,\mbox{and}\,\,\,p\in\mathbb{Z}.
	\end{equation}
\end{lemma}
\begin{definition}
	\label{psi1}
	Let $p\in\mathbb{Z}$, $\omega\in\mathbb{R}_+$ and $g^p=\{g^{p}_k\}_{k=0}^{+\infty}$, where $g^{p}_k$ is from Lemma \ref{cel1}, then we define the set of functions,  corresponding to the Laplace operator $A$, by the equality
	\begin{equation}
		\Psi_p^{\omega}=\Psi_p^{\omega}[g^p],\,\,\,\,\,\,p\in\mathbb{Z}.
	\end{equation}
\end{definition}

As it is noted earlier, the Seeley--DeWitt coefficient $a_k$ is defined for $k\geqslant0$. Let us expand this definition to the instance $k\in\mathbb{Z}$ by the equality $a_k=0$ for $k<0$. In this case, relations (\ref{rec}) are satisfied for all integer values of $k$.

\begin{lemma}
	\label{le}
	Functions from Definition \ref{psi1} have alternative representation in the form
	\begin{equation}\label{Psi}
		\Psi_p^{\omega}=\Delta^{1/2}\sum_{k=0}^{+\infty}\frac{(-1)^k\omega^{k}}{k!2^{k}}a_{p+k},\,\,\,\,\,\,p\in\mathbb{Z}.
	\end{equation}
\end{lemma}
\noindent\textbf{Proof:}
The statement follows from definition (\ref{def2}), Lemma \ref{cel1}, and the remark that $g^p_k=0$ for $k<p$. $\square$

It follows from the last lemma that the function $\Psi_p^{\omega}$ does not contain negative powers of  $\omega$ in its series representation. This means that we can consider any value of $p$ from $\mathbb{Z}$.

\begin{corollary}
	\label{cor}
	Let $d$ be even, then from Lemma \ref{th} we have $A\Psi_{d/2-1}^{\sigma}=0$.
\end{corollary}

Two important conclusions can be drawn from Lemma \ref{th}. Firstly, in the case of odd dimensions none of the functions $\Psi_p^{\sigma}$ lies entirely in the kernel of the Laplace operator. This leads to a doubly infinite chain of functions on which the Laplace operator is a shift operator. Hence, we have
\begin{equation}
	\label{cep1}
	\ldots\xrightarrow{A}\frac{\Psi_{k-1}^{\sigma}}{\Gamma(d/2-k+1)}
	\xrightarrow{A}\frac{\Psi_{k}^{\sigma}}{\Gamma(d/2-k)}
	\xrightarrow{A}\ldots,
\end{equation}
where $k\in\mathbb{Z}$ and $d$ is odd.

Secondly, in the case of even dimensions, in view of Corollary \ref{cor}, we obtain that the chain is interrupted, and we get two separate pieces of doubly infinite sequence
\begin{align}
	\label{cep}
	\ldots\xrightarrow{A}\Psi_{d/2-2}^{\sigma}
	\xrightarrow{A}\Psi_{d/2-1}^{\sigma}
	\xrightarrow{A}\,&\,0
	\\&\nonumber
	\Psi_{d/2}^{\sigma}
	\xrightarrow{A}-\Psi_{d/2+1}^{\sigma}
	\xrightarrow{A}\ldots
\end{align}

Formula (\ref{Psi_rec_1}) also allows the visual representation in the form of a sequence. However, in this case, the dimension does not matter and we get
\begin{equation}
	\label{cep3}
	\ldots\xrightarrow{-2\partial_\omega}\Psi_{k-1}^{\omega}
	\xrightarrow{-2\partial_\omega}\Psi_{k}^{\omega}
	\xrightarrow{-2\partial_\omega}\Psi_{k+1}^{\omega}
	\xrightarrow{-2\partial_\omega}\ldots,
\end{equation}
where $k\in\mathbb{Z}$.

\subsection{The case $\alpha\in\mathbb{Z}+1/2$} 
Let us consider the second useful example and extend Definition \ref{psi1}.
\begin{lemma}
	\label{cel2}
	Let $g^{1/2}_0=\sqrt{2\pi}$, then, taking into account the equalities from (\ref{soot1}), we have
	\begin{equation}
		\label{cel3}
		g^p_k=(-1)^k\Gamma(p-k)2^{p}\,\,\,\mbox{for all}\,\,\,k\geqslant 0\,\,\,\mbox{and}\,\,\,p\in\mathbb{Z}+1/2.
	\end{equation}
\end{lemma}
\begin{definition}
	\label{psi2}
	Let $p\in\mathbb{Z}+1/2$, $\omega\in\mathbb{R}_+$ and $g^{p}=\{g^{p}_k\}_{k=0}^{+\infty}$, where $g^{p}_k$ is from Lemma \ref{cel2}. Then we extend Definition \ref{psi1} to the set of half-integer indices by the relation
	\begin{equation}
		\Psi_p^\omega=\Psi_p^\omega[g^p],\,\,\,\,\,\,p\in\mathbb{Z}+1/2.
	\end{equation}
\end{definition}
Taking into account definition (\ref{def2}), we conclude, that negative powers of $\sigma$ may occur in the case of half-integer indices. Hence, according to Lemma \ref{th}, relation (\ref{Psi_rec}) holds only for $\alpha<d/2-1$. In other cases, we need to use generalized functions with the support at the point $y$. Anyway, equality (\ref{Psi_rec}) holds for the points $x\in U\setminus\{y\}$.

\begin{lemma}
	\label{lem3}
	Let $d$  be odd, then, given  (\ref{de}), we obtain $A\Psi_{d/2-1}^\sigma/(4\pi)^{d/2}=\mathbf{1}$.
\end{lemma}
\noindent\textbf{Proof:}
Due to the presence of relation (\ref{Psi_rec}) we have to check the equality only for the singular part with $\sigma^{1-d/2}$. Then, from the formula
\begin{equation}
	A_0(x)\frac{\Delta^{1/2}(x,y)}{(4\pi)^{d/2}}\sqrt{\pi}\sigma^{1-d/2}(x,y)
	=g^{-1/2}(x)\delta(x-y),
\end{equation}
where $A_0(x)=-g^{-1/2}(x)\partial_{x^{\mu}}g^{1/2}(x)g^{\mu\nu}(x)\partial_{x^{\nu}}$,
the statement of the lemma follows. $\square$

Thus, in the case of half-integer index values we can also construct a sequence on which the Laplace operator is a shift operator. Such a chain terminates on one side and has the following form
\begin{equation}
	\label{cep2}
	\ldots\xrightarrow{A}\frac{1}{2}\Psi_{d/2-3}^\sigma
	\xrightarrow{A}\Psi_{d/2-2}^\sigma
	\xrightarrow{A}\Psi_{d/2-1}^\sigma
	\xrightarrow{A}(4\pi)^{d/2}\mathbf{1}.
\end{equation}

The sequence, on which the derivative by the parameter $\sigma$ is a shift operator, is similar to (\ref{cep3}) with the only change that $k\in\mathbb{Z}\to k\in\mathbb{Z}+1/2$.

\section{Odd-dimensional case}
\label{sec:odd}
\subsection{The main result}
\label{sec:odd:res}
\begin{theorem}
	\label{th2}
	Let $K(\tau)$ be the heat kernel expansion (\ref{K_0}) on an odd-dimensional manifold for small enough values of $\tau$, $\Psi$-functions be from Definitions \ref{psi1} and \ref{psi2}, and $\omega>0$. Then, under the general conditions of Section \ref{sec:pro}, we have
	\begin{equation}
		\label{thh}
		K(\tau)e^{-\frac{\omega-\sigma}{2\tau}}=-\int_{\mathbb{R}_+} ds\,\sqrt{\frac{\tau}{s}}J_1(2\sqrt{\tau s})\Bigg(
		\sum_{n=1}^{+\infty}\frac{f_n(s)}{(4\pi)^{d/2}}\Psi_{d/2-n}^\omega+
		\sum_{n\in\mathbb{Z}}\frac{g_n(s)}{(4\pi)^{d/2}}\Psi_{(d-1)/2-n}^\omega
		\Bigg),
	\end{equation}
	where
	\begin{align}
		\label{fun}
		f_n(s)&=\sum_{k=1}^{n}\frac{(-s)^k(-m^2)^{n-k}}{k!(k-1)!\Gamma(n-k+1)},\\\nonumber
		g_{n}(s)&=\frac{\pi(-1)^{n-1}m^{2n-1}s}{\Gamma(n+1/2)}{}_1\mathrm{F}_1(1/2-n,2;-s/m^2),
	\end{align}
	and ${}_1\mathrm{F}_1$ is the confluent hypergeometric function of the first kind.
\end{theorem} 
\noindent\textbf{Proof:}
We start with representation (\ref{K_0}). Let us note that we can investigate only one-dimensional case, because a transition to the $d$-dimensional one can be achieved by applying the operator $-(2\pi)^{-1}\partial_{\omega}$ several times. Namely, $(d-1)/2$ times. It follows from formula (\ref{Psi_rec_1}) of Lemma \ref{th} and the equality
\begin{equation}
	\frac{1}{(4\pi\tau)^{d/2}}e^{-\omega/2\tau}=\left(-\frac{1}{2\pi}\frac{\partial}{\partial \omega}\right)^{\frac{d-1}{2}}\frac{1}{(4\pi\tau)^{1/2}}e^{-\omega/2\tau}.
\end{equation}

Hence, without loss of generality we consider only $d=1$. The statement can be achieved in several steps. First of all let us rewrite the left hand side of (\ref{thh}) as
\begin{equation}
	\frac{\Delta^{1/2}}{(4\pi\tau)^{1/2}}e^{-\omega/2\tau}e^{-m^2\tau}\sum_{n=0}^{+\infty}\tau^na_n=
	\frac{\Delta^{1/2}}{(4\pi\tau)^{1/2}}e^{-\omega/2\tau}\sum_{n=0}^{+\infty}a_n(-\partial_{m^2})^ne^{-m^2\tau}.
\end{equation}

Then we note one valuable relation
\begin{equation}\label{i}
	\frac{1}{(4\pi\tau)^{1/2}}e^{-\frac{\omega}{2\tau}}e^{-m^2\tau}=\frac{1}{2\pi}\int_{\mathbb{R}} d\rho\, e^{i\rho\sqrt{2\omega}}e^{-(\rho^2+m^2)\tau},
\end{equation}
that follows from applying the Fourier transform. 

Let us transform the exponential $\exp(-(\rho^2+m^2)\tau)$ in formula (\ref{i}) by applying relation (\ref{g19}) with the parameter $\lambda=\rho^2+m^2$. As a result, we get
\begin{equation}
	K(\tau)e^{-\frac{\omega-\sigma}{2\tau}}=-\int_{\mathbb{R}_+}ds\,\sqrt{\frac{\tau}{s}} J_1(2\sqrt{\tau s})
	\bigg[
	\frac{\Delta^{1/2}}{2\pi}
	\sum_{n=0}^{+\infty} a_n(-\partial_{m^2})^n
	\int_{\mathbb{R}} d\rho\, e^{i\rho\sqrt{2\omega}}
	(e^{-s/(\rho^2+m^2)}-1)
	\bigg].
\end{equation}

Thereupon, let us note that the exponential in the last formula can be expanded in the series, because $m>0$ and all integrals converge. Therefore, using the following relations
\begin{equation}
	\label{odd1}
	e^{-s/(\rho^2+m^2)}-1=
	\sum_{k=1}^{+\infty}\frac{(-s)^k}{k!(k-1)!}
	(-\partial_{m^2})^{k-1}\frac{1}{\rho^2+m^2},\,\,\,
	\frac{1}{2m}e^{-\sqrt{2m^2\omega}}=
	\frac{1}{2\pi}\int_{\mathbb{R}}d\rho\,e^{i\rho\sqrt{2\omega}}\frac{1}{\rho^2+m^2},
\end{equation}
the expression in the square brackets can be rewritten in the form
\begin{equation}
	\label{odd}
	\Delta^{1/2}
	\sum_{k=1}^{+\infty}\frac{(-s)^k(-\partial_{m^2})^{k-1}}{k!(k-1)!}
	\Bigg(
	\sum_{n=0}^{+\infty} a_n(-\partial_{m^2})^n
	\frac{e^{-\sqrt{2m^2\omega}}}{2m}
	\Bigg).
\end{equation}

Actually, we need to investigate only the term with $k=1$ in formula (\ref{odd}), because other terms can be obtained by differentiating with respect the parameter $m$. We notice that our task is a combinatorial one. In order to find a formula, we perform the following procedure. We set the parameter $\omega$ equal to the Synge's world function $\sigma$. Then, our construction will satisfy an additional equation (see below). After that we find a formula for the special case. And then, knowing this, we set the parameter $\sigma$ equal to $\omega$ and check, that the obtained formula gives the initial expansion. 

Let us note that we have obtained the series by integer degrees of the parameter $s$.
It is equal to $\hat{N}(s)=N(s)-\mathbf{1}$, see formulae (\ref{eqAN}) and (\ref{gl9}). Hence, it satisfies the equation $(\partial_s(A+m^2)+1)\hat{N}(s)=-\mathbf{1}$. If we look for a solution in the form of a series $\hat{N}(s)=\sum_{k=1}^{+\infty}(-s)^kG_k/k!$, we will find a system of recurrence relations for $G_k$. Luckily, we are interested only in the first one, that has the following form $(A+m^2)G_1=\mathbf{1}$.
For this reason, we can find the solution as a series by $\Psi$-functions with unknown coefficients.
Let us use the fact and take an ansatz in the form
\begin{equation}
	\Delta^{1/2}
	\sum_{n=0}^{+\infty} a_n(-\partial_{m^2})^n
	\frac{1}{2\sqrt{m^2}}e^{-\sqrt{2m^2\sigma}}
	=\sum_{n=1}^{+\infty}\frac{b_n}{\sqrt{4\pi}}\Psi_{1/2-n}^\sigma+
	\sum_{n\in\mathbb{Z}}\frac{c_n}{\sqrt{4\pi}}\Psi_{-n}^\sigma.
\end{equation}

Then, applying the operator $A+m^2$, see equation (\ref{Psi_rec}), and equating the answer to $\mathbf{1}$, we get two recurrence relations

\begin{equation}
	b_{n+1}=-\frac{m^2}{n-1}b_n\,\,\,\mbox{for}\,\,\,n\geqslant 1,\,\,\,
	\mbox{and}\,\,\,
	c_{n+1}=-\frac{m^2}{n+1/2}c_n\,\,\,\mbox{for}\,\,\,n\in\mathbb{Z}.
\end{equation}

This means we need to calculate only two numbers $b_1$ and $c_0$. Moreover, we can limit ourselves to calculating the coefficients near $a_0$. Therefore, using the decomposition $\exp(-\sqrt{2m^2\sigma})/2m=1/2m-\sqrt{\sigma/2}+\ldots$, we get $b_1=1$ and $c_0=\sqrt{\pi/m^2}$. Hence, we obtain

\begin{equation}
	b_{n}=\frac{(-m^2)^{n-1}}{\Gamma(n)},\,\,\,\,\,\,
	c_{n+1}=\pi\frac{(-1)^{n}(m^2)^{n-1/2}}{\Gamma(n+1/2)}.
\end{equation}

Then, using two equalities
\begin{equation}
	(-\partial_m)^{k-1}\sum_{n=1}^{+\infty}\frac{(-m^2)^{n-1}}{\Gamma(n)}\Psi_{1/2-n}^\sigma=
	\sum_{n=k}^{+\infty}\frac{(-m^2)^{n-k}}{\Gamma(n-k+1)}\Psi_{1/2-n}^\sigma,
\end{equation}
\begin{equation}
	(-\partial_m)^{k-1}\sum_{n\in\mathbb{Z}}\frac{\pi(-1)^{n}(m^2)^{n-1/2}}{\Gamma(n+1/2)}\Psi_{-n}^\sigma=\sum_{n\in\mathbb{Z}}\frac{\pi(-1)^{n-k+1}(m^2)^{n-k+1/2}}{\Gamma(n-k+3/2)}\Psi_{-n}^\sigma,
\end{equation}
we can rewrite formula (\ref{odd}) in the following form
\begin{equation}
	\label{we}
	\frac{1}{\sqrt{4\pi}}
	\sum_{k=1}^{+\infty}\frac{(-s)^k}{k!(k-1)!}
	\Bigg(
	\sum_{n=k}^{+\infty}\frac{(-m^2)^{n-k}}{\Gamma(n-k+1)}\Psi_{1/2-n}^\sigma+
	\sum_{n\in\mathbb{Z}}\frac{\pi(-1)^{n-k+1}(m^2)^{n-k+1/2}}{\Gamma(n-k+3/2)}\Psi_{-n}^\sigma
	\Bigg).
\end{equation}

Further, changing the order of summation and using the relation
\begin{equation}
	\sum_{k=1}^{+\infty}\frac{(s/m^2)^k}{k!(k-1)!\Gamma(n-k+3/2)}=\frac{s/m^2}{\Gamma(n+1/2)}
	{}_1\mathrm{F}_1(1/2-n,2;-s/m^2),
\end{equation}
we obtain the statement of the theorem. Let us note it again, that if we set the parameter $\sigma$ equal to $\omega$ in the $\Psi$-functions and expand (\ref{we}) in a series, we will get exactly (\ref{odd}). $\square$

\subsection{Examples}
\label{ex1}
Let us consider some informative calculations. We put $\omega=\sigma$. The key fact in the proof of Theorem \ref{th2} was the inequality $m>0$, thanks to which we have used the series expansion (\ref{odd1}). Now we study the limit transition $m\to+0$ for a special case. Indeed, the functions from (\ref{fun}) have the following behavior near zero 
\begin{equation}
	\lim_{m\to+0}f_n(s)=\frac{(-s)^n}{n!(n-1)!}\,\,\,\mbox{for}\,\,\,n\geqslant 0,
	\,\,\,\mbox{and}\,\,\,
	\lim_{m\to+0}g_{n}(s)=\frac{\pi(-1)^{n-1}s^{n-1/2}}{\Gamma(n+1/2)\Gamma(n+3/2)}
	\,\,\,\mbox{for}\,\,\,n\in\mathbb{Z}.
\end{equation}

This leads to the appearance of negative powers of the variable $s$ in the second equality for $n<1$. Hence, the decomposition at zero is destroyed. This means that either we have to use a different ansatz or impose additional restrictions. Let us choose the second way and consider the case, when $a_n=\delta_{n0}$. We also put $\Delta^{1/2}=1$ for convenience.

Firstly, we consider one-dimensional situation. In the case we can compute the integral in (\ref{i}) for $m=0$ explicitly. It is equal to
\begin{equation}
	\label{QQ}
	\frac{1}{2\pi}\int_{\mathbb{R}} d\rho\, e^{i\rho\sqrt{2\sigma}}(e^{-s/\rho^2}-1)=Q_1(s,\sigma)+Q_2(s,\sigma),
\end{equation}
where
\begin{equation}
	\label{Q1}
	Q_1(s,\sigma)=s\bigg(\frac{\sigma}{2}\bigg)^{\frac{1}{2}} {}_0\mathrm{F}_2\bigg(;\frac{3}{2},2;\frac{s\sigma}{2}\bigg)=\frac{1}{(4\pi)^{1/2}}
	\sum_{k=1}^{+\infty}\frac{(-s)^k}{k!}
	\bigg(\frac{\sigma}{2}\bigg)^{k-1/2}\frac{\Gamma(1/2-k)}{\Gamma(k)},
\end{equation}
and
\begin{equation}
	\label{Q2}
	Q_2(s,\sigma)=-\bigg(\frac{s}{\pi}\bigg)^{\frac{1}{2}} {}_0\mathrm{F}_2\bigg(;\frac{1}{2},\frac{3}{2};\frac{s\sigma}{2}\bigg)=-\frac{\sqrt{\pi s}}{2}\sum_{k=0}^{+\infty}\frac{(s\sigma/2)^k}{k!\Gamma(k+1/2)\Gamma(k+3/2)}.
\end{equation}

Then, using the last formulae, it is easy to check that
\begin{equation}
	\sum_{k=1}^{+\infty}\frac{f_k(s)}{\sqrt{4\pi}}\Psi_{1/2-k}^\sigma\bigg|_{a_n=\delta_{n0},\,m=0}=Q_1(s,\sigma),\,\,\,
	\sum_{k\in\mathbb{Z}}\frac{g_k(s)}{\sqrt{4\pi}}\Psi_{-k}^\sigma
	\bigg|_{a_n=\delta_{n0},\,m=0}=Q_2(s,\sigma).
\end{equation}

The last functions and their sum are depicted in Figure \ref{Fig1} for $\sigma=1/2$. It shows that $Q_1(s,1/2)$ increases to $+\infty$, when $s\to+\infty$. Therefore, we can not integrate it. We should investigate the sum $Q_1+Q_2$, because it goes to zero. The last sum is depicted in Figure \ref{Fig2} for different values of the parameter $\sigma$. As seen, the function has  damped oscillations near the abscissa axis. Actually, the damping, when $s\to+\infty$, can be obtained from formula (\ref{QQ}) explicitly by using the integration by parts.

\begin{figure}[h]
	
	\begin{minipage}[h]{0.48\linewidth}
		\center{\includegraphics[width=1.1\linewidth]{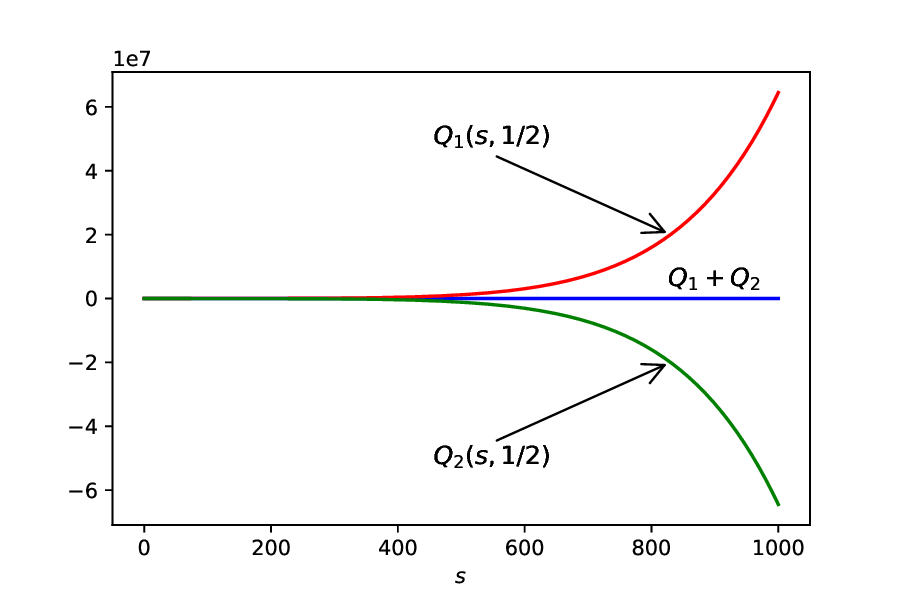}}
		\caption{The functions $Q_1(s,\sigma)$, $Q_2(s,\sigma)$, and their sum, for fixed $\sigma=1/2$, where  $1e7=10^7$.}
		\label{Fig1}
	\end{minipage}
	\hfill
	\begin{minipage}[h]{0.48\linewidth}
		\center{\includegraphics[width=1.1\linewidth]{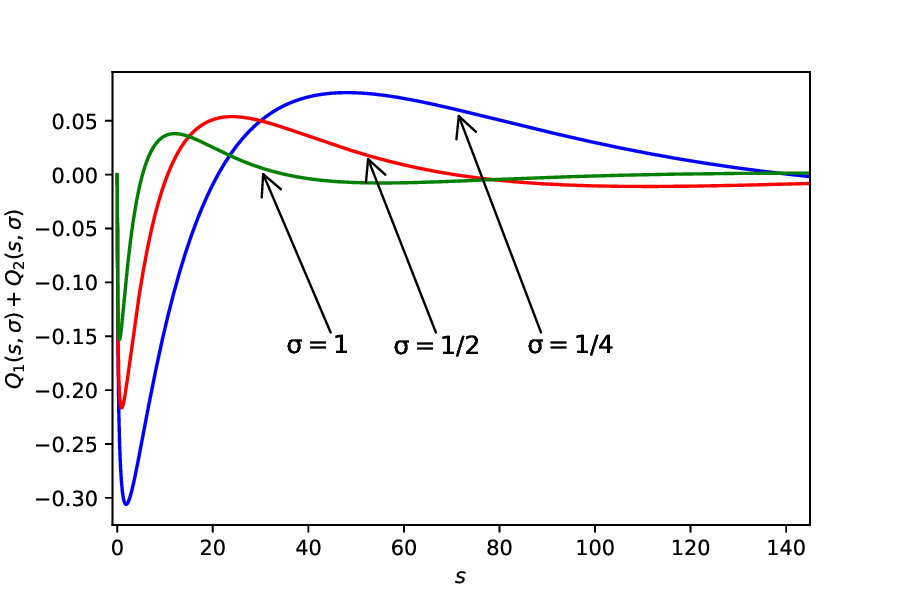}}
		\caption{The sum $Q_1(s,\sigma)+Q_2(s,\sigma)$ for different values of the parameter $\sigma=1,1/2,1/4$.}
		\label{Fig2}
	\end{minipage}
\end{figure}
\begin{figure}[h]
	\begin{center}
		\begin{minipage}[h]{0.48\linewidth}
			\center{\includegraphics[width=1.1\linewidth]{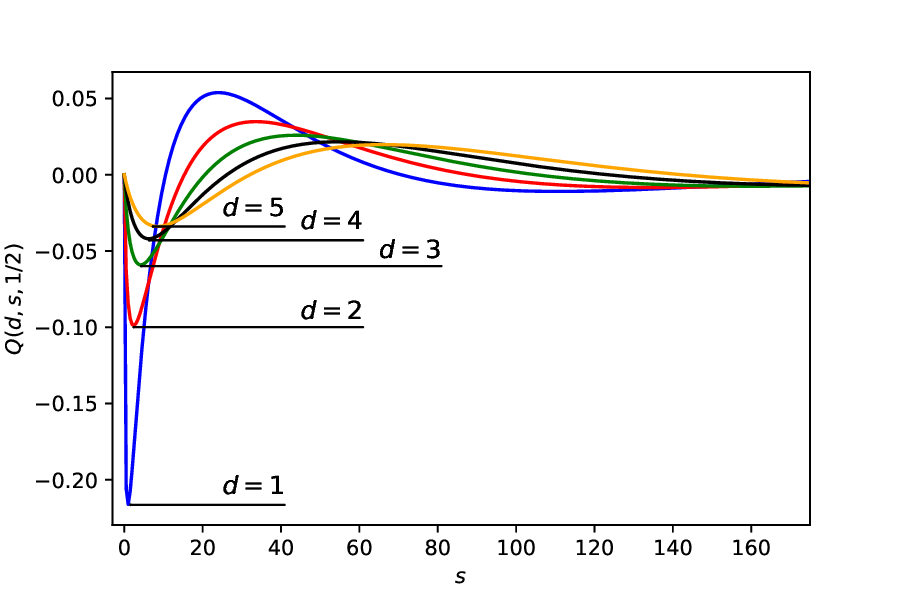}}
			\caption{The function $Q(d,s,\sigma)$ for different values of the parameter $d=1,2,3,4,5$ and fixed $\sigma=1/2$.}
			\label{Fig3}
		\end{minipage}
	\end{center}
\end{figure}

The situation $d=1$ is the case, when the Green's function has a finite limit for $\sigma=0$, because the asymptotic starts from $\sqrt{\sigma}$. But the consideration of the limit $\sigma\to+0$ in equality (\ref{thh}) is not quite correct, since  the oscillations become very large (see Figure \ref{Fig2}). So, permutation of the limit $\sigma\to+0$ and the integration by $s$ is not correct.

The situation $d>1$ can be achieved by differentiation with respect the parameter $\sigma$. So we obtain for odd values of the parameter $d$ the following formula
\begin{multline}
	\label{d}
	Q(d,s,\sigma)=\bigg(-\frac{\partial_\sigma}{2\pi}\bigg)^{\frac{d-1}{2}}[Q_1(s,\sigma)+Q_2(s,\sigma)]\\=
	\frac{1}{(4\pi)^{d/2}}
	\sum_{k=1}^{+\infty}\frac{(-s)^k}{k!}
	\bigg(\frac{\sigma}{2}\bigg)^{k-d/2}\frac{\Gamma(d/2-k)}{\Gamma(k)}\\
	-\frac{\sqrt{\pi s}}{2}\sum_{k=0}^{+\infty}\frac{(-2\pi)^{(1-d)/2}(s/2)^k\sigma^{k-(d-1)/2}}{\Gamma(k-(d+1)/2)\Gamma(k+1/2)\Gamma(k+3/2)}.
\end{multline}

The results for different odd values of the parameter $d$ are depicted in Figure \ref{Fig3}.

\section{Even-dimensional case}
\label{sec:even}
\subsection{Family of $\Phi$-functions}
\label{sec:even:rep}

Let us find the asymptotic expansion of a local solution for the problem $A\tilde{G}=\mathbf{1}$ in $U$. We notate it by $\tilde{G}=\tilde{G}(x,y)$. For the odd dimensional case the asymptotic expansion of the solution at $x\sim y$ is dictated by Lemma \ref{lem3}, that is $\tilde{G}=\Psi_{d/2-1}^\sigma/(4\pi)^{d/2}$. In the case of even $d$ the asymptotic expansion is more intricate and contains a logarithmic component. This is due to the fact that the chain (\ref{cep}) is interrupted. 

Let $d$ be even, then it is known, see Ref. \cite{15}, that for $x\sim y$ the asymptotic expansion has the form
\begin{equation}
\label{anzgreen}
\tilde{G}=\frac{\Delta^{1/2}}{(4\pi)^{d/2}}\sum_{k=0}^{d/2-2}\Gamma(d/2-k-1)
\big(\sigma/2\big)^{k+1-d/2}a_k-\frac{\big[\ln \sigma\big]}{(4\pi)^{d/2}}\Psi_{d/2-1}^\sigma+Q,
\end{equation}
where $Q=Q(x,y)$ is a smooth part.
\begin{lemma}
	\label{l1}
	Under the conditions described above $Q$ from (\ref{anzgreen}) has the form
	\begin{equation}
	\label{Q}
	Q=\frac{\Delta^{1/2}}{(4\pi)^{d/2}}\sum_{k=0}^{+\infty} \frac{(-\sigma/2)^{k+1}H_{k+1}}{(k+1)!}a_{d/2+k}+\frac{\Theta_1}{(4\pi)^{d/2}},
	\end{equation}
	where $H_k$ is the $k$-th harmonic number and $\Theta_1=\Theta_1(x,y)$ is a solution of
	$A\Theta_1=\Psi_{d/2}^\sigma$.
\end{lemma}
\noindent\textbf{Proof:} Let us substitute the ansatz (\ref{anzgreen}) into equality (\ref{eqgreen}) and use relations (\ref{det}), (\ref{rec}), (\ref{soot}), and
\begin{equation}
\label{log1}
A_0[\ln\sigma]=\frac{1}{\sigma}[A_0\sigma]+\frac{2}{\sigma},
\end{equation}
where $A_0(x)=-g^{-1/2}(x)\partial_{x^\mu} g^{1/2}(x)g^{\mu\nu}(x)\partial_{x^\nu}$.
Then, after some reductions, we obtain the equality
\begin{equation*}
AQ=\frac{\Delta^{1/2}}{(4\pi)^{d/2}}\sum_{k=0}^{+\infty}\frac{(-1)^k\sigma^k}{(k+1)!2^k}\left(d/2+2k+1+\sigma^\mu D_\mu\right)a_{d/2+k}.
\end{equation*}

The same result can be achieved after calculation of $AQ$ by using representation (\ref{Q}). This fact proves the validity of the lemma statement.
$\square$

\begin{lemma}
	\label{th1}
	Let $\omega\in\mathbb{R}_+$, $\tilde{g}^{k}=\{\tilde{g}_n^{k}\}_{n=0}^{+\infty}$, $k\in\mathbb{Z}$, are sets of complex numbers. Then we define a function of two variables $x,y\in U$, corresponding to the operator $A$, by the formula
	\begin{equation}
	\label{log}
	\Phi_{k}^{\omega}[\tilde{g}^{k}]=-[\ln\omega]\Psi_{k}^\omega+
	\Delta^{1/2}\sum_{n=0}^{+\infty}
	\big(\omega/2\big)^{n-k}\tilde{g}^{k}_n a_n^{\phantom{k}}.
	\end{equation}
	
	We assume that the numbers $\tilde{g}_n^k$ are defined by the following equalities
	\begin{equation}
	\label{ot11}
	\tilde{g}^k_n=\frac{(-1)^{n-k}H_{n-k}}{\Gamma(n-k+1)}=
	\begin{cases}
	\frac{(-1)^{n-k}H_{n-k}}{\Gamma(n-k+1)}, &\mbox{for}\,\,\,n\geqslant k;\\
	\,\,\,\,\,\Gamma(k-n), &\mbox{for}\,\,\,n< k,
	\end{cases}
	\end{equation}
	where in the first equality we have used the analytic continuation of the Gamma function and harmonic number. Let $d$ be even.
	Then the functions $\Phi_{d/2-k}^\omega\big[\tilde{g}^{d/2-k}\big]$ for $k\geqslant1$ satisfy the relations
	\begin{equation}
	\label{gl3}
	A\big(\Phi_{d/2-k-1}^\sigma\big[\tilde{g}^{d/2-k-1}\big]+H_{k}\Psi_{d/2-k-1}^\sigma\big)=
	k\big(\Phi_{d/2-k}^\sigma\big[\tilde{g}^{d/2-k}\big]+H_{k-1}\Psi_{d/2-k}^\sigma\big),
	\end{equation}
	\begin{equation}
	\label{gl31}
	A\Phi_{d/2-1}^\sigma\big[\tilde{g}^{d/2-1}\big]=(4\pi)^{d/2}\mathbf{1}-\Psi_{d/2}^\sigma,
	\end{equation}
	and
	\begin{equation}
	\label{gl32}
	-2\partial_\omega\Phi_{k}^\omega\big[\tilde{g}^{k}\big]=\Phi_{k+1}^\omega\big[\tilde{g}^{k+1}\big]\,\,\,\mbox{for all}\,\,\,k\in\mathbb{Z}.
	\end{equation}
\end{lemma}
\noindent\textbf{Proof:} The proof of relation (\ref{gl3}) for $k\geqslant1$ can be achieved be
explicit calculations and using equalities (\ref{det}), (\ref{rec}), (\ref{Psi_rec}), (\ref{soot}), and (\ref{log1}). The second formula follows from Lemma \ref{l1}. The last relation can be obtained by differentiation with taking into account the equality
\begin{equation}
(k-n)\tilde{g}_n^k+(-1)^n2^{-k}g_n^k=\tilde{g}_n^{k+1},
\end{equation}
that follows from formula (\ref{cel}) and definition (\ref{ot11}).
$\square$

\begin{definition}
	\label{log2}
	Let $d$ be even, $\omega\in\mathbb{R}_+$, $k\in\mathbb{Z}$, and the set of coefficients $\tilde{g}^k=\{\tilde{g}_n^k\}_{n=0}^{+\infty}$ be from Lemma \ref{th1}. Then we define a function, corresponding to the Laplace operator $A$, by the formula
	\begin{equation}
	\Phi_k^\omega=\Phi_{k}^\omega\big[\tilde{g}^k\big].
	\end{equation}
\end{definition}

Therefore, we have two new sequences:
\begin{equation}
\label{se1}
\ldots\xrightarrow{A}\frac{1}{2!}(\Phi_{d/2-3}^\sigma+\frac{3}{2}\Psi_{d/2-3}^\sigma)
\xrightarrow{A}(\Phi_{d/2-2}^\sigma+\Psi_{d/2-2}^\sigma)
\xrightarrow{A}\Phi_{d/2-1}^\sigma
\xrightarrow{A}(4\pi)^{d/2}\mathbf{1}-\Psi_{d/2}^\sigma,
\end{equation}
\begin{equation}
\label{se2}
\ldots\xrightarrow{-2\partial_\omega}\Phi_{d/2-3}^\omega
\xrightarrow{-2\partial_\omega}\Phi_{d/2-2}^\omega
\xrightarrow{-2\partial_\omega}\Phi_{d/2-1}^\omega
\xrightarrow{-2\partial_\omega}\ldots
\end{equation}

\begin{corollary}
	Let $d$ be even, $k\geqslant1$, and $\tilde{g}^k=\{\tilde{g}_n^k\}_{n=0}^{+\infty}$ be a set of coefficients, such that relations (\ref{gl3}) hold. Then we have $\Phi_{d/2-k}^\sigma\big[\tilde{g}^{d/2-k}\big]=\Phi_{d/2-k}^\sigma+\sum_{n=0}^{k-1}\alpha_n^k\Psi_{d/2-1-n}^\sigma$, where $\alpha_n^k\in\mathbb{C}$.
\end{corollary}
\noindent\textbf{Proof:}
It is enough to consider the difference $\Phi_{d/2-k}^\sigma\big[\tilde{g}^{d/2-k}\big]-\Phi_{d/2-k}^\sigma$ and use two facts. Firstly, negative degrees of $\sigma$ are fixed uniquely by formula (\ref{gl3}). Secondly, the difference contains only positive degrees of $\sigma$. Hence, the statement follows from Lemma \ref{th}.
$\square$
\begin{corollary}Let $d$ be even, $\omega\in\mathbb{R}_+$, and $k\in\mathbb{Z}$, then we have
	\begin{equation}
	\Delta^{\frac{1}{2}}\sigma^\mu D_\mu\Delta^{-\frac{1}{2}}\Phi_{k}^{\omega}=
	-A\Phi_{k-1}^{\omega}-k\Phi_{k}^{\omega}+
	\frac{\omega}{2}\Phi_{k+1}^{\omega}-\Psi_k^\omega.
	\end{equation}
\end{corollary}
\noindent\textbf{Proof:}
It is enough to use formulae (\ref{rec}), (\ref{Psi_rec_2}), (\ref{gl32}), and the property $\tilde{g}_n^k=\tilde{g}_{n+1}^{k+1}$ that follows from (\ref{ot11}).
$\square$
\begin{corollary}
	Let $d$ be even, $\tilde{G}$ and $\Theta_1$ be from Lemma \ref{l1}, then we have $\Phi_{d/2-1}^\sigma+\Theta_{1}=(4\pi)^{d/2}\tilde{G}$.
\end{corollary}

To conclude this section, let us give an example of an explicit construction of $\Theta_1$.
\begin{lemma}
	Let the group of functions $\{f_k=f_{k}(x,y)\}_{k=0}^{+\infty}$ satisfy the following system of equations
	\begin{equation}
	\label{mm}
	f_{0}=0,\,\,\,
	-a_{d/2+k}+(d/2+k+\sigma^\mu D_\mu)f_{k+1}=-\Delta^{-1/2}A\Delta^{1/2}f_{k},\,\,\, k\geqslant 0.
	\end{equation}
	Then the series 
	$V=\Delta^{1/2}\sum_{k=1}^{+\infty}(-\sigma/2)^kf_{k}/k!$
	solves the equation $AV=\Psi_{d/2}^\sigma$.
\end{lemma}
\noindent\textbf{Proof:}
The statement of the lemma can be verified by explicit calculations with use of equations (\ref{mm}).
$\square$

\subsection{The main result}
\label{sec:even:res}

\begin{lemma}
	\label{lem2}
	Let $d$ be even and $\tau$ be small enough, then the series 
	$\Omega^\omega(\tau)=\sum_{n=0}^{+\infty} \tau^n \Psi_{d/2+n}^\omega$
	for $\omega=\sigma$
	gives a solution for the system
	\begin{equation}
	(\partial_\tau+A)\Omega^\sigma(\tau)=0,\,\,\,\Omega(0)=\Psi_{d/2}^\sigma.
	\end{equation}
\end{lemma}
\noindent\textbf{Proof:}
The statement follows from equality (\ref{Psi_rec}).
$\square$

\begin{corollary}
	Let $\omega\in\mathbb{R}_+$ and $\tau$ be small enough. Then,
	under the conditions of Lemma \ref{lem2}, we have
	\begin{equation}\label{K_wave}
	\frac{\Omega^\omega(\tau)}{(4\pi)^{d/2}}=
	\frac{\Delta^{1/2}}{(4\pi\tau)^{d/2}}\sum_{n=d/2}^{+\infty} a_n\tau^n\sum_{k=0}^{n-d/2}\frac{(-\omega/2\tau)^k}{k!}.
	\end{equation}
\end{corollary}
\noindent\textbf{Proof:}
The statement follows from combination of Lemmas \ref{le} and \ref{lem2}.
$\square$

\begin{theorem}
	\label{th3}
	Let $\omega\in\mathbb{R}_+$, $\widetilde{K}^\omega(\tau)=K(\tau)e^{-(\omega-\sigma)/2\tau}-e^{-\tau m^2}\Omega^\omega(\tau)/(4\pi)^{d/2}$, see formulae (\ref{K_0}) and (\ref{K_wave}), on the even-dimensional manifold for small enough values of $\tau$, and $\Phi$-functions and $\Psi$-functions be from Definitions \ref{psi1} and \ref{log2}. Then, under the general conditions of Section \ref{sec:pro}, we have
	\begin{equation}
	\label{tth}
	\widetilde{K}^\omega(\tau)=-\int_{\mathbb{R}_+} ds\, \sqrt{\frac{\tau}{s}}J_1(2\sqrt{\tau s})\Bigg(\sum_{n=1}^{+\infty}\frac{f_n(s)}{(4\pi)^{d/2}}\Phi_{d/2-n}^\omega+\sum_{n=1}^{+\infty} \frac{g_{-n}(s)}{(4\pi)^{d/2}}\Psi_{d/2-n}^\omega\Bigg),
	\end{equation}
	where
	\begin{equation}
	\label{rec1}
	f_{n}(s)=\sum_{k=1}^{n}\frac{(-s)^k(-m^2)^{n-k}}{k!(k-1)!\Gamma(n-k+1)},
	\end{equation}
	and
	\begin{align}
	\label{rec2}
	g_{-n}(s)=&\sum_{k=1}^{n}\frac{(-s)^k(-m^2)^{n-k}(-2\gamma-\ln(m^2/2))}{k!(k-1)!\Gamma(n-k+1)}\\&
	+(-m^2)^{n-1}\frac{\partial}{\partial\epsilon}\bigg|_{\epsilon=0}\frac{se^{-\gamma\epsilon}}{\Gamma(n+\epsilon)}{}_1\mathrm{F}_1(1-n-\epsilon,2;-s/m^2),
	\end{align}
	where $\gamma$ is the Euler--Mascheroni constant.
\end{theorem}
\noindent\textbf{Proof:}
First of all we simplify the problem. Namely, due to the presence of equalities (\ref{Psi_rec_1}) and (\ref{gl32}) we can investigate only two dimensional case. Indeed, the situations $d>2$ can be obtained by using differentiation by the parameter $\omega$. Hence, without loss of generality we work only with $d=2$.

Let us start from the right hand side of (\ref{tth}). Using formulae (\ref{K_0}) and (\ref{K_wave}), $\widetilde{K}^\omega(\tau)$ can be rewritten in the form
\begin{equation}
\widetilde{K}^\omega(\tau)=\frac{\Delta^{1/2}}{4\pi\tau}e^{-\tau m^2}\sum_{n=0}^{+\infty}\tau^n a_n\left(e^{-\omega/2\tau}-\sum_{k=0}^{n-1}\frac{1}{k!}\left(-\frac{\omega}{2\tau}\right)^k\right).
\end{equation}

Then, we introduce an integration operator $S(\omega)$ and its inverse one, such that $S^{-1}(\omega)S(\omega)=1$, by the formulae
\begin{equation}
S(\omega)= -\frac{1}{2}\left.\int_0^\omega \, d\widetilde{\omega}\right|_{\omega=\widetilde{\omega}},\,\,\,\,\,\,
S^{-1}(\omega)=-2\frac{d}{d\omega}.
\end{equation}

In this case the formula for $\widetilde{K}^\omega(\tau)$ has the following form
\begin{equation}\label{K_wave_2}
\widetilde{K}^\omega(\tau)=\frac{\Delta^{1/2}}{4\pi\tau}e^{-\tau m^2}\sum_{n=0}^{+\infty} a_n\left(S^{n}(\omega)e^{-\omega/2\tau}\right), 
\end{equation}
or, using relation (\ref{g19}) and two-dimensional Fourier transform for $(4\pi\tau)^{-1}\exp(-\omega/2\tau-m^2\tau)$, as it was done in the one-dimensional case earlier, we get
\begin{equation}\label{K_wave_3}
\widetilde{K}^\omega(\tau)=-\int_{\mathbb{R}_+} ds\,\sqrt{\frac{\tau}{s}}J_1(2\sqrt{\tau s})
\bigg[
\frac{\Delta^{1/2}}{4\pi}\sum_{n=0}^{+\infty} a_nS^{n}(\omega)\,2
\int_{\mathbb{R}_+} d\rho\,\rho J_0(\rho\sqrt{2\omega})\left(e^{-s/(\rho^2+m^2)}-1\right)\bigg],
\end{equation}
where we have used one additional equality
\begin{equation}
\label{int}
\frac{1}{\pi}\int_{\mathbb{R}_+}\int_0^{2\pi} d\rho d\phi\,\rho e^{i\rho\sqrt{2\omega}\cos(\phi)}\left(e^{-s/(\rho^2+m^2)}-1\right)=
2\int_{\mathbb{R}_+} d\rho\, \rho J_0(\rho\sqrt{2\omega})\left(e^{-s/(\rho^2+m^2)}-1\right).
\end{equation}

Now we need to transform the expression in the square brackets to the $\Psi$-($\Phi$-)functions representation. For this purpose we can use the series expansion for the exponential, because all integrals converge due to $m>0$. Let us use the following two relations
\begin{equation}
\label{exp}
e^{-s/(\rho^2+m^2)}-1=\sum_{k=1}^{+\infty} \frac{(-s)^k}{k!(k-1)!}\left(-\frac{\partial}{\partial m^2}\right)^{k-1}\frac{1}{\rho^2+m^2},
\end{equation}
and
\begin{equation}
\label{exp1}
\int_{\mathbb{R}_+} d\rho\,\frac{\rho J_0(\rho\sqrt{2\omega})}{\rho^2+m^2}=
\mathrm{K}_0(\sqrt{2m^2\omega})
=\sum_{j=0}^{+\infty}\frac{H_j-\gamma-\ln\sqrt{m^2\omega/2}}{j!j!}\left(\frac{m^2\omega}{2}\right)^j,
\end{equation}
where $\mathrm{K}_0$ is the modified Bessel function of second kind.
Thus, we need to find new representation for the expression
\begin{equation}\label{U}
\sum_{k=1}^{+\infty} \frac{(-s)^k}{k!(k-1)!}\left(-\frac{\partial}{\partial m^2}\right)^{k-1}\Bigg[\sum_{n=0}^{+\infty} \frac{a_nS^n(\omega)}{2\pi}\Bigg(\sum_{j=0}^{+\infty}\frac{H_j-\gamma-\ln\sqrt{m^2\omega/2}}{j!j!}\left(\frac{m^2\omega}{2}\right)^j\Bigg)\Bigg].
\end{equation}

Indeed, we need to work only with the coefficient $k=1$, because other can be obtained by differentiation by the parameter $m^2$. Here we must give arguments similar to those that were in Theorem \ref{th2}, so we give the link to the paragraph after formula (\ref{odd}) to avoid repetition. Taking into account the previous reasoning, analyzing the terms from the sum leads to the ansatz $\sum_{n=1}^{+\infty} b_n(\Phi_{1-n}^\sigma+H_{n-1}\Psi_{1-n}^\sigma)/(4\pi)+\sum_{n=1}^{+\infty} c_n\Psi_{1-n}^\sigma/(4\pi)$, that we should choose. Moreover, we can use the fact, that the coefficient for $k=1$ gives $4\pi\mathbf{1}-\Psi_{1}^\sigma$ after applying the operator $A+m^2$ to it. Hence, we obtain two sets of recurrence relations
\begin{equation}
b_{n+1}=-\frac{m^2}{n}b_n\,\,\,\mbox{and}\,\,\,
c_{n+1}=-\frac{m^2}{n}c_n\,\,\,\mbox{for}\,\,\,n\geqslant 1.
\end{equation}

The last equalities mean, that we need to find only initial coefficients $b_1$ and $c_1$. They follow from formula (\ref{Psi}), definition (\ref{log}), and the coefficients from (\ref{U}) corresponding to the term $n=j=0$. Hence, we get $b_1=1$ and $c_1=-2\gamma-\ln(m^2/2)$, from which the main formulae follow
\begin{equation}
b_n=\frac{(-m^2)^{n-1}}{(n-1)!}\,\,\,\mbox{and}\,\,\, c_n=\frac{(-m^2)^{n-1}}{(n-1)!}[-2\gamma-\ln(m^2/2)]
\,\,\,\mbox{for}\,\,\,n\geqslant 1.
\end{equation}

Substituting this into (\ref{U}) and differentiating by the parameter $m^2$, as it was done in the one-dimensional case, we obtain
\begin{equation}
\label{re}
\frac{1}{4\pi}
\sum_{k=1}^{+\infty}\frac{(-s)^k}{k!(k-1)!}\Bigg(\sum_{n=k}^{+\infty} \frac{(-m^2)^{n-k}}{\Gamma(n-k+1)}\big[\Phi_{1-n}^\sigma-\Psi_{1-n}^\sigma(2\gamma+\ln(m^2/2))\big]+\sum_{n=1}^{+\infty}\frac{(-m^2)^{n-k}H_{n-k}}{\Gamma(n-k+1)}\Psi_{1-n}^\sigma\Bigg),
\end{equation}
where we have used the following relations
\begin{equation}
\frac{\partial_\epsilon\Gamma(n+\epsilon+1)}{\Gamma(n+\epsilon+1)}=H_{n+\epsilon}-\gamma,
\end{equation}
\begin{equation}
\partial_t^k t^n\ln(t)=\partial_t^k\partial_\epsilon^{\phantom{k}}\big|_{\epsilon=0} t^{n+\epsilon}=
\partial_\epsilon^{\phantom{k}}\big|_{\epsilon=0} \frac{t^{n-k+\epsilon}\Gamma(n+\epsilon+1)}{\Gamma(n-k+\epsilon+1)}=
\frac{t^{n}\Gamma(n+1)}{\Gamma(n-k+1)}(\ln(x)+H_n-H_{n-k}).
\end{equation}

Then, changing the summation order and using the following formulae
\begin{equation}
\frac{H_{n-k}}{\Gamma(n-k+1)}=-\frac{\partial}{\partial\epsilon}\bigg|_{\epsilon=0}
\frac{e^{-\gamma\epsilon}}{\Gamma(n-k+1+\epsilon)},
\end{equation}
\begin{equation}
\sum_{n=1}^{+\infty}\frac{t^kH_{n-k}}{k!(k-1)!\Gamma(n-k+1)}=
-\frac{\partial}{\partial\epsilon}\bigg|_{\epsilon=0}\frac{te^{-\gamma\epsilon}}{\Gamma(n+\epsilon)}{}_1\mathrm{F}_1(1-n-\epsilon,2;-t),
\end{equation}
we obtain the theorem statement. Let us note it again, that if we put $\sigma$ equal to $\omega$ and expand (\ref{re}) in a series, we will get exactly (\ref{U}).
$\square$

\begin{lemma}
	\label{th4}
	Let $\omega\in\mathbb{R}_+$, $\Omega^\omega(\tau)=\sum_{n=0}^{+\infty} \tau^n \Psi_{d/2+n}^\omega$ be from Lemma \ref{lem2}, and $\tau$ be small enough. Then, under the general conditions of Section \ref{sec:pro}, we have
	\begin{equation}
	\label{tth2}
	e^{-\tau m^2}\Omega^\omega(\tau)=-\int_{\mathbb{R}_+} ds\, \sqrt{\frac{\tau}{s}}J_1(2\sqrt{\tau s}) \left(\sum_{n=0}^{+\infty} g_n(s)\Psi_{d/2+n}^\omega\right),
	\end{equation}
	where
	\begin{equation}
	g_n(s)=(-\partial_{m^2})^{n}\big[e^{-s/m^2}-1\big].
	\end{equation}
\end{lemma}
\noindent\textbf{Proof:}
The statement of the lemma follows from the explicit expression for $\Omega(\tau)$, relation \begin{equation}
\tau^n\exp(-\tau m^2)=(-\partial_{m^2})^{n}\exp(-\tau m^2)\,\,\,\mbox{for}\,\,\,n\in\mathbb{N},
\end{equation}and formula (\ref{g19}) for $\lambda=m^2$.
$\square$

\begin{corollary}
	\label{th5}
	Let $K(\tau)$ be the heat kernel expansion (\ref{K_0}) on the even-dimensional manifold for small enough values of $\tau$, $\omega\in\mathbb{R}_+$, and $\Phi$-functions and $\Psi$-functions be from Definitions \ref{psi1} and \ref{log2}. Let also the functions $\{f_n(s),g_{-n}(s)\}_{n=1}^{+\infty}$ be from Theorem \ref{th3}, and $\{g_{n}(s)\}_{n=0}^{+\infty}$ be from Lemma \ref{th4}. Then, under the general conditions of Section \ref{sec:pro}, we have
	\begin{equation}
	\label{tth1}
	K(\tau)e^{-\frac{\omega-\sigma}{2\tau}}=-\int_{\mathbb{R}_+} ds\, \sqrt{\frac{\tau}{s}}J_1(2\sqrt{\tau s})\Bigg(\sum_{n=1}^{+\infty}\frac{f_n(s)}{(4\pi)^{d/2}}\Phi_{d/2-n}^\omega+\sum_{n\in\mathbb{Z}} \frac{g_{n}(s)}{(4\pi)^{d/2}}\Psi_{d/2+n}^\omega\Bigg).
	\end{equation}
\end{corollary}
\noindent\textbf{Proof:}
The statement of the corollary follows from the sum of formulae (\ref{tth}) and (\ref{tth2}).
$\square$

\subsection{Examples}
\label{ex2}
Here we continue to consider examples for $\omega=\sigma$. In the previous section we have derived the Hankel transform of heat kernel expansion (\ref{K_0}) in the even-dimensional case, see Theorem \ref{th3}, Lemma \ref{th4}, and Corollary \ref{th5}. An important detail in the proof was the positive mass $m>0$. Thanks to this fact we were able to use the series expansion (\ref{exp}), because integral (\ref{exp1}) converges. 

However, the transition $m\to+0$ to the massless case is not an easy task, because the transform of $\Omega(\tau)$, see Lemma \ref{th4}, contains the exponential $\exp(-s/m^2)$. Hence, permutation of the limit $m\to+0$ and integration by $s$ is not quite correct. Therefore, to study the limiting case we need to either use different decomposition or consider a special type restriction. Let us choose the second way and investigate the case $a_n=\delta_{n0}$ and $\Delta^{1/2}=1$. Moreover, firstly we study the simplest situation $d=2$.

The condition $a_n=\delta_{n0}$ means, that $\Omega(\tau)$ is equal to zero. So, the functions $\{g_n(s)\}_{n\geqslant 0}$ are zero. Then we can find the asymptotic behavior for other functions
\begin{equation}
\lim_{m\to+0}f_{n}(s)=\frac{(-s)^n}{n!(n-1)!},\,\,\,
\lim_{m\to+0}g_{-n}(s)=\frac{(-s)^n}{n!(n-1)!}(\ln(2)-\ln(s)-3\gamma+H_n+H_{n-1}),
\end{equation}
for $n\geqslant 1$, and where we have used (\ref{rec1}), (\ref{rec2}), and the relation
\begin{equation}
{}_1\mathrm{F}_1(1-n-\epsilon,2,-s/m^2)=\frac{(s/m^2)^{n-1+\epsilon}}{\Gamma(n+1+\epsilon)},\,\,\,\mbox{for}\,\,\,m\to+0.
\end{equation}

Then, using the summation with $\Phi$- and $\Psi$- functions, we get
\begin{equation}
\sum_{n=1}^{+\infty}\frac{f_n(s)}{4\pi}\Phi_{1-n}^\sigma\bigg|_{a_n=\delta_{n0},\,m=0}=\frac{1}{4\pi}\sum_{n=1}^{+\infty} \frac{s^n(\ln\sigma-H_{n-1})}{(n-1)!(n-1)!n!}(\sigma/2)^{n-1}, 
\end{equation}
\begin{equation}
\sum_{n=1}^{+\infty} \frac{g_{-n}(s)}{4\pi}\Psi_{1-n}^\sigma\bigg|_{a_n=\delta_{n0},\,m=0}=\frac{1}{4\pi}\sum_{n=1}^{+\infty} \frac{s^n(\ln(s)-\ln(2)+3\gamma-H_n-H_{n-1})}{(n-1)!(n-1)!n!}(\sigma/2)^{n-1}.
\end{equation}

Let us note that the last sums converge, but they have bad behaviour at infinity, because they  increase exponentially. However, their sum $Q(2,s,\sigma)$ is a good function, because it goes to zero at infinity. The same situation has been studied in the odd-dimensional case. This is a rather remarkable property of series that the logarithmic growth of $\ln(s)$ can be compensated by introducing a harmonic numbers. Also we should note, that the last result in the restricted case can be obtained by explicit calculation of (\ref{int}), see Appendix A.

The case for $d>2$ can be obtained by differentiating $-\partial_\sigma/2\pi$ with respect to the parameter $d/2-1$ times. The result function $Q(d,s,\sigma)$, when $d$ is even, is equal to 
\begin{align}
Q(d,s,\sigma)&=(-\partial_\sigma/2\pi)^{d/2-1}Q(2,s,\sigma)\\&=
\frac{(-1)^{d/2-1}}{(4\pi)^{d/2}}\sum_{n=1}^{+\infty} \frac{s^n[\ln(s\sigma/2) -H_{n-1}-H_{n}-H_{n-d/2} +3\gamma]}{(n-1)!n!\Gamma(n-d/2+1)}(\sigma/2)^{n-d/2}
\end{align}
and depicted in Figure \ref{Fig3} for different values of the parameter $d$. The last formula complements the definition for odd-dimensional case (\ref{d}). Especially, we note that the decreasing of $Q(2,s,\sigma)$ for $s\to+\infty$ can be shown explicitly integrating (\ref{int}) by parts.

\section{Applications}
\label{apl}

\paragraph{Green's function expansion.} We can extract several useful corollaries from Theorem \ref{th2} and Corollary \ref{th5}. Indeed, in both statements, on the right side of the equality, we have the Hankel transform of a function that depends on the variable $s$. This function also allows a Taylor series expansion in powers of the parameter $s$. At the same time, according to the reasoning from Section \ref{sec:pro}, the coefficient near the first degree of $s$ is a Green's function. So, for $x\sim y$ we get the following formal decompositions

	\begin{equation}\label{greenodd}
	G=\frac{1}{(4\pi)^{d/2}}
	\Bigg(
	\sum_{n=1}^{+\infty}\frac{(-m^2)^{n-1}}{\Gamma(n)}\Psi_{d/2-n}^\sigma+
	\sum_{n\in\mathbb{Z}}\frac{\pi(-1)^{n}(m^2)^{n-1/2}}{\Gamma(n+1/2)}\Psi_{(d-1)/2-n}^\sigma\Bigg),
	\end{equation}
	when $d$ is odd, and
	\begin{equation}\label{greeneven}
	G=\frac{1}{(4\pi)^{d/2}}\Bigg(
	\sum_{n=1}^{+\infty} \frac{(-m^2)^{n-1}}{\Gamma(n)}\big[\Phi_{d/2-n}^\sigma-\Psi_{d/2-n}^\sigma(2\gamma+\ln(m^2/2)-H_{n-1})\big]+
	\sum_{n=1}^{+\infty}\frac{\Gamma(n)}{m^{2n}}\Psi_{d/2-1+n}^\sigma\Bigg),
	\end{equation}
	when $d$ is even.
To obtain them it is enough to use the inverse Hankel transform in equality (\ref{gl9}), apply the Taylor expansion to the left hand side and use the results from Theorem \ref{th2} and Corollary \ref{th5}.

Note that the first few terms of the last series in the four-dimensional case were written out in Ref. \cite{10}. Let us draw the attention that the last expansions should be considered as  the formal ones in powers of $x-y$. Convergence issues or asymptotic behaviour should be discussed separately. They depend on the smoothness of the Laplace operator coefficients and the mass value. Some arguments on this subject can be found in Refs. \cite{5,10,102}.

\paragraph{Cutoff regularization.} The representations for $\Psi$- and $\Phi$-functions obtained above, see formula (\ref{def2}) and definition (\ref{log}), can be useful in different investigations with a cutoff regularization. For example, in loop calculations, the main building block for diagrams is a Green's function. As a rule, it is represented as an integral of $K(x,y;\tau)$ from $0$ to $+\infty$ by the variable $\tau$. If we want to explore the region $x\sim y$ , we will face a problem, because the integral at $x=y$ diverges near zero. This means that we should use a regularization.

One of the convenient types of regularization is the cutoff one. It can be achieve by introducing a special parameter $\Lambda$ by the following substitution:
\begin{equation}
\sigma \to \sigma_{\Lambda}=
\begin{cases}
\,\,\,\,\sigma, &\text{for}\,\,\sigma\geqslant 1/2\Lambda^2;\\
1/2\Lambda^2, &\text{for}\,\,\sigma<1/2\Lambda^2.
\end{cases}
\end{equation}

After such procedure in the region $\sigma<1/2\Lambda^2$ we have the exponential $\exp(-1/4\tau\Lambda^2)$ instead of $\exp(-\sigma/2\tau)$. Hence, the integral converges for $x=y$ near $\tau=0$. Moreover, it can be verified, that the regularized Green's function goes to $G(x,y)$ when $\Lambda\to+\infty$ in the sense of generalized functions.

Let us give some examples of regularized fundamental solutions, described above:
\begin{equation}
G \to G_{\Lambda}=G\big|_{\sigma \to \sigma_{\Lambda}},
\end{equation}
where the substitution affects only the parameter $\sigma$ as an independent one. We emphasize that the regularization does not change dependence of the Seeley--DeWitt coefficients on the variables $x$ and $y$. This approach has been used successfully in the renormalization procedure of the four-dimensional Yang--Mills theory \cite{16} and the scalar $\phi^3$-model \cite{113}.

\paragraph{Integral calculations.}
We have seen above that sometimes we need to compute integrals by variable $\tau$. For example, when finding the Green's function. However, this is not the only example. Similar integrals arise in quantum field theory and in the theory of anomalies. In this regard, it makes sense to study constructions of the following form
\begin{equation}
\label{int1}
\int_{\mathbb{R}_+}d\tau\,\tau^{-k/2}K(x,y;\tau).
\end{equation}
where it is implied that the integral exists on the upper limit.

Let us consider the following procedure. We introduce an auxiliary manifold, that does not violate the general assumptions of the problem and locally can be represented by Euclidean coordinates $\hat{x}$ and $\hat{y}$ from a smooth convex domain  $\hat{U}\subset\mathbb{R}^k$. For example, product of $k$ circles. Then we define the Synge’s world function as $\hat{\sigma}(\hat{x},\hat{y})=|\hat{x}-\hat{y}|^2/2$. Hence, the local part of a heat kernel $\hat{K}(\hat{x},\hat{y};\tau)$ for the ordinary Laplace operator has the following form $(4\pi\tau)^{-k/2}\exp(-\hat{\sigma}/2\tau)$. Further we note that for $x\sim y$
\begin{equation}
\nonumber
\int_{\mathbb{R}_+}d\tau\,\tau^{-k/2}K(x,y;\tau)=
\lim_{\hat{y}\to\hat{x}}\int_{\mathbb{R}_+}\frac{d\tau}{(4\pi)^{-\frac{k}{2}}}\hat{K}(\hat{x},\hat{y};\tau)K(x,y;\tau)=
(4\pi)^{k/2}\lim_{\hat{y}\to\hat{x}}G((x,\hat{x}),(y,\hat{y})),
\end{equation}
where $G$ is the Green's function for $(d+k)$-dimensional case, in which the Seeley--DeWitt coefficients, the Van-Vleck--Morette determinant, and the Synge's world function are from the Green's function for $d$-dimensional case.

The latter formula means that we can study integrals (\ref{int1}) by using the transition into a space of higher dimension. In the last formula we have used $d\to d+k$. This trick, with further use of the representations for $\Psi$-($\Phi$-)functions, allows us to find divergences and simplify calculations, see the section with $\ln\det(A)$ from Ref. \cite{16}.

\paragraph{Integral transforms of the heat kernel.} Actually, the result of Theorem \ref{th2}, Corollary \ref{th5}, and Lemma \ref{l1} allows us to find some integral transformations for the local part of the heat kernel. Let us assume, that an integral transformation is on $\mathbb{R}_+$, and that a kernel $C(\tau)$ of the integral transform has a convergent Taylor decomposition $C(\tau)=\sum_{n=0}^{+\infty}c_n\tau^n$. Let us also remind the fact that the heat kernel $K(\tau)$ and the Green's function $G$ from (\ref{greenodd}) and (\ref{greeneven}) depend on the mass parameter $m$. So we can write the following chain of relations
\begin{equation}
\int_{\mathbb{R}_+}d\tau\,C(\tau)K(\tau)=
\sum_{n=0}^{+\infty}c_n\int_{\mathbb{R}_+}d\tau\,\tau^nK(\tau)=
\sum_{n=0}^{+\infty}c_n(-\partial_{m^2})^nG,
\end{equation}
where we changed the order of the sum and the integral under the assumption that we have convergence. Otherwise, we assume that we are working with asymptotic series. Let us also note that if the integral transform contains negative powers of the variable $\tau$, then we can use the derivatives with respect to the parameter $\sigma$, considering it as an independent one.

\section{Conclusion}
In our paper we have defined two families of new functions ($\Psi$ and $\Phi$), and have used them to find the Hankel transform of the local part of the heat kernel (\ref{K_0}). We believe that such calculations are very useful and important because they lead to a set of non-trivial relations, that allow us to study the transformations of the heat kernel. For example, the usual integration leads to the Green's function. Other applications are listed in Section \ref{apl}.

Separately, we pay attention to the fact that we perform all calculations locally, when $x,y\in U\subset M$. Moreover, in some cases we assume that $x\sim y$. This is necessary in order to be able to uniquely construct the Synge's function $\sigma(x,y)$ and the Seeley--DeWitt coefficients $a_n (x,y)$. This raises interesting questions. In what cases is it possible to abandon locality and conduct reasoning globally? 

Another problem can be related to investigation of  products of the functions. Let $x,y\in U$. Consider the following integral
\begin{equation}
h_{p,n}=\int_{U}d^dz\,\Psi_{k}^\sigma(x,z)\sqrt{g(z)}\Psi_{n}^\sigma(z,y).
\end{equation}

Of course we have no ability to calculate the integral explicitly, but we can give some discussions for special cases. Let $d$ be odd and $p=d/2-k$, where $k\in\mathbb{N}$, then if we apply the operator $A$ several times we get the delta-function under the integral. This means that $A^kh_{p,n}=\Psi_{n}^\sigma$. Another example is the following. Let $d$ be even and $p=d/2-k$, where $k\in\mathbb{N}$, then we obtain $A^kh_{p,n}=0$. The same reasoning are possible for $\Phi$-functions.

Here we make some obvious remarks on the spectrum of the operator $A+m^2$. In Section \ref{sec:pro}, we assumed that such an operator has exclusively positive eigenvalues. If we have a non-positive value of the spectral parameter, then the integral along $\mathbb{R}_+$ will diverge. Therefore, we must add the projector on the positive component of the spectrum in all our calculations.

If we assume that the manifold $M$ has a boundary, then the heat kernel must satisfy boundary conditions. Such an assumption changes the asymptotic expansion (\ref{K_0}) and it should have additional correction terms. This means that the case of manifolds with a boundary must be considered separately and is not within the scope of this paper.

Let us make some comment on dimension of physical quantities. In this paper, we have used such formulae as $\ln(m^2)$ or $\ln(\sigma)$. We meant that all the values are dimensionless. If we consider the case when $[\sigma]=L^2$ and $[m]=L^{-1}$, where $L$ has the dimension of length, this is not a problem, because all calculations contain logarithm sums. Hence, we have only dimensionless combinations of the quantities.

\paragraph{Acknowledgments.}
	This research is fully supported by the grant in the form of subsidies from the Federal budget for state support of creation and development world-class research centers, including international mathematical centers and world-class research centers that perform research and development on the priorities of scientific and technological development. The agreement is between MES and PDMI RAS from \textquotedblleft8\textquotedblright\,November 2019 № 075-15-2019-1620. 
	
	Authors express gratitude to D.V. Vassilevich for reading of the manuscript and suggestion of amendments.
	Also, A.V. Ivanov is a winner of the Young Russian Mathematician award and would like to thank its sponsors and jury.

\section{Appendix A}
\begin{lemma}
	The integral (\ref{int}) from Theorem \ref{th3} has the following series expansion near zero for $m=0$
	\begin{align}\label{int_t}
	T(\omega,s)&= 2\int_{\mathbb{R}_+} d\rho\,\rho J_0(\rho\sqrt{2\omega})\left(e^{-s/\rho^2}-1\right)\\&\nonumber=
	\sum_{k=1}^{+\infty} \frac{s^k[\ln(s\omega/2) -2H_{k-1}-H_{k} +3\gamma]}{(k-1)!(k-1)!k!}(\omega/2)^{k-1}.
	\end{align}
\end{lemma}
\noindent\textbf{Proof:}
First of all we show that the integral satisfies the following differential equation 
$2\partial_\omega^{\phantom{2}} s\partial_s^2T(\omega,s)=T(\omega,s)$. This can be achieved by explicit differentiation, integration by parts, and using the properties of the Bessel functions
\begin{equation}
\frac{d}{dx}J_0(x)=-J_1(x),\,\,\,\,\,\,\frac{1}{x}J_1(x)+\frac{d}{dx}J_1(x)=J_0(x).
\end{equation}

According to the above mentioned, let us take an ansatz for (\ref{int_t}) in the form
\begin{equation}\label{T_anz}
T(\omega,s)= g_0(\omega)+\sum_{k=1}^\infty \left(s^k\ln(s)f_k(\omega)+s^k g_k(\omega)\right),
\end{equation}
where the coefficients $ f_k(\omega)$ and $g_k(\omega)$ should be found. 
Let us apply the operator $2\partial_\omega^{\phantom{2}} s\partial_s^2$  to (\ref{T_anz}). Then we get the recurrent relations
\begin{equation}
\label{eq}
\begin{cases}
2\partial_\omega f_1(\omega)=g_0(\omega);
\\
2k(k+1)\partial_\omega f_{k+1}(\omega)=f_k(\omega);
\\
2(2k+1)\partial_\omega f_{k+1}(\omega)-2k(k+1)\partial_\omega g_{k+1}(\omega)=g_k(\omega).
\end{cases}
\end{equation}
where $k\geqslant 1$.

To  solve them, we have to find the initial conditions, using the integral representation (\ref{int_t}). One can note that $T(\omega,0)=0$, so $g_0(\omega)=0$. It is evident, that we can explicitly integrate the equations from (\ref{eq}). However, we have such arbitrariness as integration constants $\{f_k(0)\}_{k\geqslant1}$ and $\{g_k(0)\}_{k\geqslant2}$. Let us find them, using the asymptotic behavior of the construction (\ref{int}) for quite small values of $\omega$. For those purposes we cut $\mathbb{R}_+$ in two intervals $[0;1]$ and $[1;+\infty)$. So we get
\begin{align}
2\int_0^1d\rho\, \rho J_0(\rho\sqrt{2\omega})\left(e^{-s/\rho^2}-1\right)&=
\int_0^1 d\rho\, \left(e^{-s/\rho}-1\right)+o(1)\\&=s\ln(s)+s\omega-s-\sum_{k=2}^{+\infty} \frac{(-s)^k}{k!(k-1)}+o(1),
\end{align}
and
\begin{equation}
2\int_1^{+\infty} d\rho\,\rho J_0(\rho\sqrt{2\omega})\left(e^{-s/\rho^2}-1\pm\frac{s}{\rho^2}\right)
=\sum_{k=2}^{+\infty} \frac{(-s)^k}{k!(k-1)}+s\ln(\omega)+s\left(2\gamma-\ln(2)\right)+o(1).
\end{equation}

Using the last calculations and the form of ansatz (\ref{T_anz}), we get
\begin{equation}
\sum_{k=1}^{+\infty} s^k\ln(s) f_k(0)+\sum_{k=2}^{+\infty}s^kg_k(0)=s \ln(s).
\end{equation}

Therefore, $f_1(0)=1$ and $f_k(0)=g_k(0)=0$ for $k\geqslant 2$. Finally, we need to find the coefficient $g_1(\omega)$, that corresponds to $s$. This can be achieved by subtracting the logarithmic part $s\ln s$ and differentiating by the parameter $s$ with a further transition $s\to0$. So we get
\begin{align}
\nonumber
\partial_s\big|_{s=0}\bigg[2\int_{\mathbb{R}_+} d\rho\, \rho J_0(\rho\sqrt{2\omega})\left(e^{-s/\rho^2}-1\right)-s\ln(s)\bigg]=&-2\int_1^{+\infty} \frac{d\rho}{\rho}J_0(\rho\sqrt{2\omega})\\&-2\int_0^1\frac{d\rho}{\rho}\left(J_0(\rho\sqrt{2\omega})-1\right)+\gamma-1\\=&\ln(\omega)+3\gamma-\ln 2-1.
\end{align}
where in the second equality we have used the change of variable $\rho\to\rho/\sqrt{2\omega}$.
This means that $g_1(\omega)=\ln(\omega)+3\gamma-\ln(2)-1$.
Solving the recurrent relations, we find
\begin{equation}
f_k(\omega)=\frac{(\omega/2)^{k-1}}{k!(k-1)!(k-1)!},\,\,\,
g_k(\omega)=\frac{\ln(\omega)-2H_{k-1}-H_k+3\gamma-\ln 2}{k!(k-1)!(k-1)!}(\omega/2)^{k-1},
\end{equation}
which leads to the statement of the lemma.
$\square$

\end{document}